\documentclass[proof]{iasart}			

\usepackage{amsmath,amsfonts,amssymb,amsthm}
\usepackage{esint}

\theoremstyle{plain}
\newtheorem{thm}{\protect\theoremname}
  \theoremstyle{plain}
  \newtheorem{prop}[thm]{\protect\propositionname}
  \theoremstyle{plain}
  \newtheorem{lem}[thm]{\protect\lemmaname}
  \providecommand{\lemmaname}{Lemma}
  \providecommand{\theoremname}{Theorem}
  \providecommand{\propositionname}{Proposition}

\title{Volume tensor of pheasant brain compartments estimated by Fakir probe}
\shorttitle{Tensor by Fakir}
\shortauthors{Janacek J \etal}

\author[1]{Ji\v{r}\'{\i} Jan\'a\v{c}ek}
\author[2,3]{Daniel Jir\'ak} 

\email{jiri.janacek@fgu.cas.cz,
daniel.jirak@ikem.cz}

\affiliation[1]{Department of Biomathematics, Institute of Physiology, The Czech Academy of Sciences, V\'{\i}de\v nsk\'a 1083, 142 20 Prague, Czech Republic}

\affiliation[2]{MR Unit, Department of Diagnostic and Interventional Radiology, Institute for Clinical and Experimental Medicine, V\'{\i}de\v nsk\'a 1958/9, 140 21 Prague, Czech Republic}

\affiliation[3]{Institute of Biophysics and Informatics, 1st Medicine Faculty, Charles University, Prague, Czech Republic }

\abstract{The volume tensor provides robust estimate of object shape and orientation in space. The tensor is estimated from 3D data set by the Fakir probe, an interactive method using intersections of the objects boundary with a virtual lines. The method thus can be applied to objects that cannot be segmented automatically. Marking the intersections instead of segmenting the whole object reduces the workload required for obtaining sufficiently precise results. We present theoretical results on the variance of estimate of integrals by systematic sampling that enable calculation of the shape estimate precision. To demonstrate the ability of Fakir technique, we measure the changes in shape and orientation  of pheasant brain compartments during development.}


\begin{document}
\begin{paper}

\section{Introduction}

The centered second moment volume tensor can be used for robust characterization of objects shape and orientation by an equivalent ellipsoid \citep{ziegel:2015}. Current estimators of the tensors use either curvatures of triangulated surface of the object segmented from given 3D data \citep{schroederturk:2011} or stereological approach that requires random sections of the object \citep{rafati:2016}. Intersecting the object with the spatial grid, that is sparse in comparison with the lattice of voxels, we obtain sufficiently precise condensed information about the shape of the object. This approach, designed for measurement of surface area \citep{barbier:1860,cruz-orive:1997,kubinova:1998} was implemented in the Fakir method for estimation of the surface area and volume of 3D objects in data volumes obtained by 3D imaging modalities such as confocal microscopy or Computed Tomography (CT) \citep{kubinova:1999}. The principle of the estimators is based on integral geometry \citep{santalo:1976}, namely 1) the mean number of intersections of the object boundary with the grid is in direct proportion to the surface area of the object, and 2) the mean length of the grid lines inside the objects is in direct proportion to the volume of the object; where the mean is assessed with respect to the random position of the grid. The coefficients of proportionality are the products of the grid length density ($m^{-2}$) and constants, they equal to 1 and 1/2, respectively.

Variance of the estimators using randomly oriented grids can be estimated from the grid density and properties of the measured objects using methods originated from Matherons transitive theory \citep{matheron:1965}. An example is the variance of volume estimator of 3D objects using spatial grid with isotropic uniform random (IUR) position, that can be efficiently estimated from the asymptotic term proportional to the surface area of the object divided by square of the length density as the density tends to infinity \citep{janacek:2006,janacek:2008} with a constant characterizing the grid. 

Optimal choice of the line grid in Fakir method further reduces the workload. One particularly efficient grid of lines, can be found in the crystal of garnet with cubic grid where the atoms are aligned along lines in seven directions, three of them orthogonal and the other four along diagonals of the cube \citep{okeefe:1992}.

The aim of this paper is to demonstrate, that Fakir method can be used for numerical integration in evaluation of the volume tensor and that lengths and directions of semi-axes of the equivalent ellipsoid can be used for robust estimation of object shape and orientation. The precision of the semi-axes  estimated by the Fakir method depends on the spatial grid density and the object properties and we present approximate formula for its calculation. Moreover, we provide proof of asymptotic formula for variance of estimator of integrals over objects with finite perimeter \citep{ziemer:1989}.

The Fakir method was applied in developmental study of pheasant brain compartments \citep{jirak:2015} in order to estimate the changes of the bird brain divisions shapes with ontogeny. Here we apply our approach to the task of assessment of changes in shape and orientation of phaesant brain compartments during development.

\section{Material and Methods}
\subsection{Estimate of the second moment volume tensor using line segments} 

Let $K$ be a measurable subset of 3-dimensional Euclidean space $\mathbf{\mathbb{R}}^{3}$. The elements of the second moment volume Minkowski tensor $\Phi_{3,2,0}\left(K\right)=\left\{ \phi_{i,j}\right\} $
are defined by the integrals \citep{hug:2008}: 
\begin{equation}
\phi_{i,j}=\frac{1}{2}\iiintop_{K}x_{i}x_{j}dx_{1}dx_{2}dx_{3},\;i,j=1\ldots3.\label{eq:defint}
\end{equation}

We combine the $\Phi_{3,2,0}$ with center of mass and volume of the
object to obtain translation invariant centered second moment volume
tensor $\Theta\left(K\right)=\left\{ \tau_{i,j}\right\}$ \citep{ziegel:2015}:
\begin{equation}
\tau_{i,j}=\frac{2\phi_{i,j}}{V}-c_{i}c_{j},\;i,j=1\ldots3,\label{eq:center}
\end{equation}
where
\begin{equation}
c_{i}=\frac{m_{i}}{V},\;m_{i}=\iiintop_{K}x_{i}dx_{1}dx_{2}dx_{3},\;i=1\ldots3\label{eq:momint}
\end{equation}
are coordinates of centre of mass and
\begin{equation}
V=\iiintop_{K}dx_{1}dx_{2}dx_{3}\label{eq:volint}
\end{equation}
is volume of $K$.

Eigenvalues and eigenvectors of the centered tensor provide information
on the object shape and orientation, respectively. Let $\lambda_{i}$ and $v_{i}$, $i=1\ldots3$, be the eigenvalues
and eigenvectors of the centered tensor $\Theta\left(K\right)$. We characterize
the anisotropy of the set by the Procrustes anisotropy ($PA$) \citep{dryden:2009}:
\begin{equation}
PA\left(K\right)=\sqrt{\frac{3}{2}\frac{\sum_{i=1}^{3}\left(\sqrt{\lambda_{i}}-\frac{1}{3}\sum_{j=1}^{3}\sqrt{\lambda_{j}}\right)^{2}}{\sum_{i=1}^{3}\lambda_{i}}}.\label{eq:procr}
\end{equation}

The $PA$ takes values in interval from 0 up to 1. 

It is convenient to visualize the centered tensor $T\left(K\right)$
by an equivalent ellipsoid with semi-axes in direction $v_{i}$ with
lengths $s_{i}$, $i=1\ldots3$, \citep{ziegel:2015}, where:
\begin{equation}
s_{i}=\sqrt{5\lambda_{i}}.\label{eq:semi}
\end{equation}

The Fakir estimate of the tensor entries $\widetilde{\phi}_{i,j}$, $i,j=1\ldots3$, uses grid of lines $G$ with
intensity $L_{V}$ ($m^{-2}$), randomly shifted by uniform
random vector $U$:
\begin{equation}
\widetilde{\phi}_{i,j}=\frac{1}{2L_{V}}\intop_{G+U\cap K}x_{i}x_{j}dH\left(x\right),\label{eq:appint}
\end{equation}
where $H$ is 1-dimensional Hausdorff measure. The estimate is obviously
unbiased. The value of the estimate of $s_{i}$ in Eq. \ref{eq:semi} can be calculated from the
coordinates of intersections of the grid $G+U$ with the object $K$
as follows. Let the intersection of the set $K$ with the grid consist
of $N$ line segments with endpoints $\mathbf{a}_{k}=\left(a_{k,1},a_{k,2},a_{k,3}\right)$
and $\mathbf{b}_{k}=\left(b_{k,1},b_{k,2},b_{k,3}\right)$ and let
$l_{k}$ be the length of the $k$-th segment. Calculating $\widetilde{\phi}_{i,j}$ in Eq. \ref{eq:appint} as
the sum of integrals of over individual line segments gives:
\begin{equation}
\begin{array}{c}
\widetilde{\phi}_{i,j}=\frac{1}{12L_{V}}\sum_{k=1}^{N}l_{k}\left(2a_{k,i}a_{k,j}+a_{k,i}b_{k,j}+\right.\\
\left.+b_{k,i}a_{k,j}+2b_{k,i}b_{k,j}\right),\;i,j=1\ldots3.
\end{array}\label{eq:sumint}
\end{equation}

Estimates of the volume $V$, of the center of mass $c$ and of the
first volume moment $m=Vc$ are:
\begin{equation}
\widetilde{c}_{i}=\frac{\widetilde{m}_{i}}{\widetilde{V}},\;\widetilde{m}_{i}=\frac{1}{2L_{V}}\sum_{k=1}^{N}l_{k}\left(a_{k,i}+b_{k,i}\right),\label{eq:momap}
\end{equation}
\begin{equation}
\widetilde{V}=\frac{1}{L_{V}}\sum_{k=1}^{N}l_{k}.\label{eq:volap}
\end{equation}

The natural estimate of centered tensor element is then:
\begin{equation}
\widetilde{\tau}_{i,j}=\frac{2\widetilde{\phi}_{i,j}}{\widetilde{V}}-\widetilde{c}_{i}\widetilde{c}_{j}.\label{eq:centrap}
\end{equation}

Formulas in Eq. \ref{eq:sumint}, \ref{eq:momap} and \ref{eq:volap} are discrete analogues of Eq. \ref{eq:defint}, \ref{eq:momint} and \ref{eq:volint}, respectively. We estimate the equivalent ellipsoid and its Procrustes
anisotropy by plugging the eigenvalues of the estimated tensor $\widetilde{\lambda}_{i}$ into formulas in Eq. \ref{eq:procr} and \ref{eq:semi}.

\subsection{Precision of semi-axes estimate} 

The calculation precision of tensor components estimate is based on
evaluation of variance of estimate of integral of polynomials, or
more generally of estimate of covariance of simultaneous estimates
of two such integrals. Special case (when the polynomial is constant equal to $1$)
is known, because variance of the volume estimate by isotropic Fakir
probe is \citep{janacek:1999} 
\begin{equation}
\mathrm{var}\left(\widetilde{V}\left(K,U,R\right)\right)\cong C_{G}S\left(K\right)L_{V}^{-2},\label{eq:volasymp}
\end{equation}
where
\[
\widetilde{V}\left(K,U,R\right)=\frac{1}{L_{V}}\intop_{RG+U\cap K}dH\left(x\right),
\]
where $U$ is random shift, $R$ is random rotation and the grid constant
$C_{G}$ can be calculated from Fourier transform of the grid \citep{janacek:2010}. 
$C_{G}$ of the optimized Fakir grids are
\[
C_{G}=\frac{1}{8\pi^{3}}\left(3\zeta\left(Z_{2},4\right)-\frac{21}{2}\zeta\left(4\right)\right)
\]
with value $0.02707533$ for threefold grid, 
\[
C_{G}=\frac{1}{8\pi^{3}}\left(4\zeta\left(A_{2},4\right)-\frac{63}{4}\zeta\left(4\right)\right)
\]
with value $0.02453877$ for fourfold grid and 
\[
C_{G}=\frac{1}{8\pi^{3}}\left(3\zeta\left(Z_{2},4\right)+4\zeta\left(A_{2},4\right)-\right.
\]
\[
\left.-\frac{21}{8}\left(10+\sqrt{3}\right)\zeta\left(4\right)\right)
\]
with value $0.0317757$ for sevenfold grid, where 
\[
\zeta\left(Z_{2},s\right)=\sum{}_{i,j=-\infty}^{'\infty}\left(i^{2}+j^{2}\right)^{-\frac{s}{2}}
\]
is Epstein zeta function of square point grid, $\zeta\left(Z_{2},4\right)\cong6.02681$,
\[
\zeta\left(A_{2},s\right)=\sum{}_{i,j=-\infty}^{'\infty}\left(2\frac{i^{2}+ij+j^{2}}{\sqrt{3}}\right)^{-\frac{s}{2}}
\]
is Epstein zeta function of unit triangular point grid, $\zeta\left(A_{2},4\right)\cong5.78336$
and 
\[
\zeta\left(s\right)=\sum_{i=1}^{\infty}i^{-s}
\]
is Riemann zeta function, $\zeta\left(4\right)\cong1.082323$.

Generalization of the Eq. \ref{eq:covasymp} yields the asymptotic formula for covariance
of the estimates of integrals of complex functions $f_{1}\left(x\right)$
and $f_{2}\left(x\right)$:
\begin{equation}
\begin{array}{c}
\mathrm{cov}\left(\widetilde{I}\left(f_{1},U,R\right),\widetilde{I}\left(f_{2},U,R\right)\right)\cong\\
\cong C_{G}\iintop_{\partial K}f_{1}\left(x\right)\overline{f_{2}\left(x\right)}dS\left(x\right)L_{V}^{-2},
\end{array}\label{eq:covasymp}
\end{equation}
where 
\[
\widetilde{I}\left(f,U,R\right)=\frac{1}{L_{V}}\intop_{RG+U\cap K}f\left(x\right)dH\left(x\right).
\]

Validity of the asymptotic expansions Eq. \ref{eq:volasymp} and \ref{eq:covasymp} for smooth functions with bounded support, sufficiently regular 
sets and arbitrary periodic measures is established in Theorem \ref{thm:main}. 

The surface integral in Eq. \ref{eq:covasymp} can be estimated from intersections $x_{k}$
of the Fakir probe with the surface of the object as 
\[
\iintop_{\partial K}h\left(x\right)dS\left(x\right)\cong\frac{2}{L_{V}}\sum_{x_{k}\in RG+U\cap \partial K}h\left(x_{k}\right).
\]

Covariance of centered tensor components is approximated by surface
integral. By linearisation of Eq. \ref{eq:centrap} written as
\[
\tilde{\tau}_{i,j}=\frac{\tilde{\phi}_{i,j}}{\tilde{V}}-\frac{\tilde{m}_{i}\tilde{m}_{j}}{\tilde{V}^{2}},\;i,j=1\ldots3
\]
we obtain for $i,j,k,l=1\ldots3$:
\[
\mathrm{cov}\left(\widetilde{\tau}_{i,j},\widetilde{\tau}_{k,l}\right)\cong\frac{4}{\widetilde{V}^{2}}\mathrm{cov}\left(\widetilde{\phi}_{i,j},\widetilde{\phi}_{k,l}\right)-\frac{4\widetilde{\phi}_{i,j}}{\widetilde{V}^{3}}\mathrm{cov}\left(\widetilde{\phi}_{k,l},\widetilde{V}\right)
\]
\[
-\frac{4\widetilde{\phi}_{k,l}}{\widetilde{V}^{3}}\mathrm{cov}\left(\widetilde{\phi}_{i,j},\widetilde{V},\right)+\frac{4\widetilde{\phi}_{i,j}\widetilde{\phi}_{k,l}}{\widetilde{V}^{4}}\mathrm{var}\left(\widetilde{V}\right)
\]
\[
-\frac{2\widetilde{m}_{l}}{\widetilde{V}^{3}}\mathrm{cov}\left(\widetilde{\phi}_{i,j},\widetilde{m}_{k}\right)-\frac{2\widetilde{m}_{k}}{\widetilde{V}^{3}}\mathrm{cov}\left(\widetilde{\phi}_{i,j},\widetilde{m}_{l}\right)
\]
\[
-\frac{2\widetilde{m}_{j}}{\widetilde{V}^{3}}\mathrm{cov}\left(\widetilde{\phi}_{k,l},\widetilde{m}_{i}\right)-\frac{2\widetilde{m}_{i}}{\widetilde{V}^{3}}\mathrm{cov}\left(\widetilde{\phi}_{k,l},\widetilde{m}_{j}\right)
\]
\[
+\frac{2\widetilde{m}_{k}\widetilde{m}_{l}}{\widetilde{V}^{4}}\mathrm{cov}\left(\widetilde{\phi}_{i,j},\widetilde{V}\right)+\frac{2\widetilde{m}_{i}\widetilde{m}_{j}}{\widetilde{V}^{4}}\mathrm{cov}\left(\widetilde{\phi}_{k,l},\widetilde{V}\right)
\]
\[
+\frac{2\widetilde{\phi}_{i,j}\widetilde{m}_{l}}{\widetilde{V}^{4}}\mathrm{cov}\left(\widetilde{m}_{k},\widetilde{V}\right)+\frac{2\widetilde{\phi}_{i,j}\widetilde{m}_{k}}{\widetilde{V}^{4}}\mathrm{cov}\left(\widetilde{m}_{l},\widetilde{V}\right)
\]
\[
+\frac{2\widetilde{\phi}_{k,l}\widetilde{m}_{j}}{\widetilde{V}^{4}}\mathrm{cov}\left(\widetilde{m}_{i},\widetilde{V}\right)+\frac{2\widetilde{\phi}_{k,l}\widetilde{m}_{i}}{\widetilde{V}^{4}}\mathrm{cov}\left(\widetilde{m}_{j},\widetilde{V}\right)
\]
\[
-\frac{4\widetilde{\phi}_{i,j}\widetilde{m}_{k}\widetilde{m}_{l}}{\widetilde{V}^{5}}\mathrm{var}\left(\widetilde{V}\right)-\frac{4\widetilde{\phi}_{k,l}\widetilde{m}_{i}\widetilde{m}_{j}}{\widetilde{V}^{5}}\mathrm{var}\left(\widetilde{V}\right)
\]
\[
+\frac{\widetilde{m}_{j}\widetilde{m}_{l}}{\widetilde{V}^{4}}\mathrm{cov}\left(\widetilde{m}_{i},\widetilde{m}_{k}\right)+\frac{\widetilde{m}_{j}\widetilde{m}_{k}}{\widetilde{V}^{4}}\mathrm{cov}\left(\widetilde{m}_{i},\widetilde{m}_{l}\right)
\]
\[
+\frac{\widetilde{m}_{i}\widetilde{m}_{l}}{\widetilde{V}^{4}}\mathrm{cov}\left(\widetilde{m}_{j},\widetilde{m}_{k}\right)+\frac{\widetilde{m}_{i}\widetilde{m}_{k}}{\widetilde{V}^{4}}\mathrm{cov}\left(\widetilde{m}_{j},\widetilde{m}_{l}\right)
\]
\[
-\frac{2\widetilde{m}_{j}\widetilde{m}_{k}\widetilde{m}_{l}}{\widetilde{V}^{5}}\mathrm{cov}\left(\widetilde{m}_{i},\widetilde{V}\right)-\frac{2\widetilde{m}_{i}\widetilde{m}_{k}\widetilde{m}_{l}}{\widetilde{V}^{5}}\mathrm{cov}\left(\widetilde{m}_{l},\widetilde{V}\right)
\]
\[
-\frac{2\widetilde{m}_{i}\widetilde{m}_{j}\widetilde{m}_{l}}{\widetilde{V}^{5}}\mathrm{cov}\left(\widetilde{m}_{k},\widetilde{V}\right)-\frac{2\widetilde{m}_{i}\widetilde{m}_{j}\widetilde{m}_{k}}{\widetilde{V}^{5}}\mathrm{cov}\left(\widetilde{m}_{j},\widetilde{V}\right)
\]
\[
+\frac{4\widetilde{m}_{i}\widetilde{m}_{j}\widetilde{m}_{k}\widetilde{m}_{l}}{\widetilde{V}^{6}}\mathrm{var}\left(\widetilde{V}\right)
\]
and it follows from Eq. \ref{eq:covasymp} by proper grouping of factors that
\[
\mathrm{cov}\left(\widetilde{\tau}_{i,j},\widetilde{\tau}_{k,l}\right)\cong
\]
\[
\frac{C_{G}}{\widetilde{V}^{2}L_{V}^{2}}\iintop_{\partial K}\left(\left(x_{i}-\widetilde{c}_{i}\right)\left(x_{j}-\widetilde{c}_{j}\right)-\widetilde{\tau}_{i,j}\right)
\]
\[
\left(\left(x_{k}-\widetilde{c}_{k}\right)\left(x_{l}-\widetilde{c}_{l}\right)-\widetilde{\tau}_{k,l}\right)dS\left(x\right).
\]

Variance of semiaxes length is calculated from $\mathrm{cov}\left(\widetilde{\tau}_{i,j},\widetilde{\tau}_{k,l}\right)$
using linear approximations of formulas for eigenvalues and of Eq. \ref{eq:semi}.

Characteristic polynomial $P\left(\lambda\right)$ and invariants
$T$, $Q$, $D$ of the tensor $\Theta$ are related with $\tau_{ij}$
by formula:
\[
P\left(\lambda\right)=\left|\begin{array}{ccc}
\tau_{11}-\lambda & \tau_{12} & \tau_{13}\\
\tau_{12} & \tau_{22}-\lambda & \tau_{23}\\
\tau_{13} & \tau_{23} & \tau_{33}-\lambda
\end{array}\right|=
\]
\[
=-\lambda^{3}+T\lambda^{2}-Q\lambda+D.
\]

Partial derivatives of the invariants are then
\[
\frac{\partial T}{\partial\tau_{ii}}=1,\;\frac{\partial T}{\partial\tau_{ij}}=0,
\]
\[
\frac{\partial Q}{\partial\tau_{ii}}=\tau_{jj}+\tau_{kk},\;\frac{\partial Q}{\partial\tau_{ij}}=-2\tau_{ij},
\]
\[
\frac{\partial D}{\partial\tau_{ii}}=\tau_{jj}\tau_{kk}+\tau_{jk}^{2},\;\frac{\partial D}{\partial\tau_{ij}}=2\tau_{ik}\tau_{jk}-2\tau_{ij}\tau_{kk},
\]
where $\left\{ i,j,k\right\} =\left\{ 1,2,3\right\} .$

Partial derivatives of centered tensor eigenvalues are calculated
solving the equations with partial derivatives
of invariants:
\[
\frac{\partial\lambda_{1}}{\partial\tau_{ij}}+\frac{\partial\lambda_{2}}{\partial\tau_{ij}}+\frac{\partial\lambda_{3}}{\partial\tau_{ij}}=\frac{\partial T}{\partial\tau_{ij}},
\]
\[
\left(\lambda_{2}+\lambda_{3}\right)\frac{\partial\lambda_{1}}{\partial\tau_{ij}}+\left(\lambda_{1}+\lambda_{3}\right)\frac{\partial\lambda_{2}}{\partial\tau_{ij}}+\left(\lambda_{1}+\lambda_{2}\right)\frac{\partial\lambda_{3}}{\partial\tau_{ij}}=\frac{\partial Q}{\partial\tau_{ij}},
\]
\[
\lambda_{2}\lambda_{3}\frac{\partial\lambda_{1}}{\partial\tau_{ij}}+\lambda_{1}\lambda_{3}\frac{\partial\lambda_{2}}{\partial\tau_{ij}}+\lambda_{1}\lambda_{2}\frac{\partial\lambda_{3}}{\partial\tau_{ij}}=\frac{\partial D}{\partial\tau_{ij}}.
\]

The derivatives of eigenvalues are then:
\[
\frac{\partial\lambda_{k}}{\partial\tau_{ij}}=\left(\lambda_{l}-\lambda_{m}\right)\left(\lambda_{k}^{2}\frac{\partial T}{\partial\tau_{ij}}+\lambda_{k}\frac{\partial Q}{\partial\tau_{ij}}+\frac{\partial D}{\partial\tau_{ij}}\right)Det^{-1},
\]
where $\left(k,l,m\right)$ is $\left(1,2,3\right)$, $\left(2,3,1\right)$
or $\left(3,1,2\right)$ and 
\[
Det=\lambda_{1}^{2}\left(\lambda_{2}-\lambda_{3}\right)+\lambda_{2}^{2}\left(\lambda_{3}-\lambda_{1}\right)+\lambda_{3}^{2}\left(\lambda_{1}-\lambda_{2}\right).
\]

Variance of eigenvalues $\widetilde{\lambda}_{m}$, $m=1\ldots3$
is:
\[
\mathrm{var}\left(\widetilde{\lambda}_{m}\right)\cong\sum_{\begin{array}{c}
i,j,k,l=1\\
i\leq j,k\leq l
\end{array}}^{3}\frac{\partial\lambda_{m}}{\partial\tau_{ij}}\frac{\partial\lambda_{m}}{\partial\tau_{kl}}\mathrm{cov}\left(\widetilde{\tau}_{i,j},\widetilde{\tau}_{k,l}\right).
\]

Finally, the variance of semiaxes lengths $\widetilde{s}_{m}$, $m=1\ldots3$
is:
\begin{equation}
\mathrm{var}\left(\widetilde{s}_{m}\right)\cong\frac{5}{4\widetilde{\lambda}_{m}}\mathrm{var}\widetilde{\lambda}_{m}.\label{eq:semivar}
\end{equation}

The line probe was implementated in home made Fakir program for operation system Windows including the calculations of semiaxes length precision.

\subsection{Precision of integral estimate}

We study properties of estimate of integral by randomly oriented and
shifted periodic grid in arbitrary dimension. Variance of the estimate
is calculated using Fourier analysis and Wiener Tauberian theorem.

Basic notions and properties concerning the Fourier transform and
convolution, summarized in what follows, can be found e.g. in \cite{bochner:1949}.

For $f\in\mathbf{L}^{1}\left(\mathbf{\mathbb{R}}^{d}\right)$ the
Fourier transform $\mathcal{F}f$ is the function
\[
\mathcal{F}f\left(\xi\right)=\intop_{\mathbf{\mathbb{R}}^{d}}f\left(x\right)\exp\left(-2\pi ix\xi\right)dx.
\]

The coefficient $-2\pi i$ is replaced by $2\pi i$ in inverse Fourier
transform. 

For reflection $\widehat{f}\left(x\right)=f\left(-x\right)$ it is $\mathcal{F}\widehat{f}=\overline{\mathcal{F}f}$.

The convolution of $f_{1}$ and $f_{2}\in\mathbf{L}^{1}\left(\mathbf{\mathbb{R}}^{d}\right)$
is defined by 
\[
f_{1}\star f_{2}\left(x\right)=\intop_{\mathbf{\mathbb{R}}^{d}}f_{1}\left(x-y\right)f_{2}\left(y\right)dy.
\]

$\mathcal{F}\left(f_{1}\star f_{2}\right)=\mathcal{F}f_{1}\mathcal{F}f_{2}$
by convolution theorem for Fourier transform. 

For rotation $M\in\mathbf{S}\mathbf{O}_{d}$ (special orthogonal group)
and $f$ is 
\[
Mf\left(x\right)=f\left(M^{-1}x\right).
\]

Fourier transform of spherically symmetric function $f$ (i.e. $Mf=f$
for any $M$) is spherically symmetric. We define $f\left(\left\Vert x\right\Vert \right)=f\left(x\right)$
and $\mathcal{F}f\left(\left\Vert \xi\right\Vert \right)=\mathcal{F}f\left(\xi\right)$,
then $r^{d-1}f\left(r\right)\in\mathbf{L}^{1}\left(\mathbb{R}^{+}\right)$
and $\mathcal{F}f$ is Hankel transform
\begin{equation}
\mathcal{F}f\left(\rho\right)=2\pi\rho^{1-\frac{d}{2}}\intop_{0}^{\infty}r^{\frac{d}{2}}J_{\frac{d}{2}-1}\left(2\pi\rho r\right)f\left(r\right)dr,\label{eq:Haenkel}
\end{equation}
where $J_{\frac{d}{2}-1}$ is Bessel function of first kind. Eq. \ref{eq:Haenkel} is inverse Hankel transform as well.

For $f\in\mathbf{L}^{1}\left(\mathbf{\mathbb{R}}^{d}\right)\cap\mathbf{L}^{2}\left(\mathbf{\mathbb{R}}^{d}\right)$
the covariogram $g_{f}$ is function
\[
g_{f}\left(x\right)=\intop_{\mathbf{\mathbb{R}}^{d}}f\left(y+x\right)\overline{f\left(y\right)}dy.
\]
Obviously $g_{f}\left(-x\right)=\overline{g_{f}\left(x\right)}$.
By convolution and Parseval theorems $\mathcal{F}g_{f}=\left|\mathcal{F}f\right|^{2}\in\mathbf{L}^{1}\left(\mathbf{\mathbb{R}}^{d}\right)$
and $g_{f}=\mathcal{F}^{-1}\mathcal{F}g_{f}$.

Isotropic covariogram $\mathfrak{g}_{f}$ is real function 
\[
\mathfrak{g}_{f}\left(\left|x\right|\right)=\int_{\mathbf{S}\mathbf{O}_{d}}g_{Mf}\left(x\right)dp\left(M\right),
\]
where $p\left(M\right)$ is invariant probabilistic measure on $\mathbf{S}\mathbf{O}_{d}$.
Then 
\[\mathcal{F}\mathfrak{g}_{f}\left(\left|\xi\right|\right)=\int_{\mathbf{S}\mathbf{O}_{d}}\mathcal{F}g_{Mf}\left(\xi\right)dp\left(M\right)=
\]
\[
=\int_{\mathbf{S}\mathbf{O}_{d}}\left|\mathcal{F}Mf\right|^{2}\left(\xi\right)dp\left(M\right).
\]

Obviously $\mathcal{F}\mathfrak{g}_{f}\geq0$, $\rho^{d-1}\mathcal{F}\mathfrak{g}_{f}\left(\rho\right)\in\mathbf{L}^{1}\left(\mathbf{\mathbb{R}}^{+}\right)$
and 
\begin{equation}
\mathcal{F}\mathcal{F}\mathfrak{g}_{f}=\mathfrak{g}_{f}.\label{eq:cov-inv}
\end{equation}

Let $\mathbf{T}$ be point lattice $A\mathbf{\mathbb{Z}}^{d}$,
where $A\in\mathbf{\mathbb{R}}^{d\times d}$ is regular matrix. Fundamental
region of $\mathbf{T}$ is $F_{\mathbf{T}}=A\left[0,1\right)^{d}$
with volume $\mathrm{det}\:A$. Spatial intensity of $\mathbf{T}$
is $\alpha=\left(\mathrm{det}\:A\right)^{-1}$. We define dual lattice
$\mathbf{T}^{*}$ of points $A^{-1}\mathbf{\mathbb{Z}}^{d}$ .

$\mathbf{T}$-periodic measure $\mu$ is Borel $\sigma$-finite measure
in $\mathbf{\mathbb{R}}^{d}$ such that $\mu\left(K+x\right)=\mu\left(K\right)$
for all $x\in\mathbf{T}$ and all Borel set $K$. Let $\lambda$ be
the intensity of $\mu$ equal to $\alpha\mu\left(F_{\mathbf{T}}\right)$.
Fourier coefficient $\mu$ with index $\xi\in\mathbf{T}^{*}$ is 
\[
m_{\xi}=\alpha\intop_{F_{\mathbf{T}}}\exp\left(-2\pi ix\xi\right)d\mu\left(x\right).
\]

Obviously $m_{0}=\lambda$.

Convolution of $\sigma$-finite Borel measure $\mu$ in $\mathbf{\mathbb{R}}^{d}$
with function $f\in\mathbf{L}^{1}\left(\mathbf{\mathbb{R}}^{d}\right)$
is 
\[
f\star\mu\left(x\right)=\intop_{\mathbf{\mathbb{R}}^{d}}f\left(x-y\right)d\mu\left(y\right).
\]

Let $\mu$ be $\mathbf{T}$-periodic measure and let $f\in\mathbf{L}^{1}\left(\mathbf{\mathbb{R}}^{d}\right)\cap\mathbf{L}^{2}\left(\mathbf{\mathbb{R}}^{d}\right)$.
Then in the space $\mathbf{L}^{2}\left(F_{\mathbf{T}},\alpha\lambda^{d}\right)$ 
\[
f\star\mu\left(x\right)=\sum_{\xi\in\mathbf{T}^{*}}m_{\xi}\mathcal{F}f\left(\xi\right)\exp\left(2\pi ix\xi\right)
\]
and the mean and variance of $f\star\mu$ are
\begin{equation}
\begin{array}{c}
\alpha\intop_{F_{\mathbf{T}}}f\star\mu\left(x\right)dx=m_{0}\mathcal{F}f\left(0\right)=\\
=\lambda\intop_{\mathbf{\mathbb{R}}^{d}}f\left(x\right)dx\equiv\lambda J_{f},
\end{array}\label{eq:mean}
\end{equation}
\begin{equation}
\begin{array}{c}
\alpha\intop_{F_{\mathbf{T}}}\left|f\star\mu\left(x\right)-\lambda J_{f}\right|^{2}dx=\\
=\sum_{\xi\in\mathbf{T}^{*}}^{\xi\neq0}\left|m_{\xi}\right|^{2}\left|\mathcal{F}f\left(\xi\right)\right|^{2}.
\end{array}\label{eq:var-freq}
\end{equation}

Let $\mu$ be $\mathbf{T}$-periodic measure and $u\in\mathbf{\mathbb{R}}^{+}$.
We define $u\mathbf{T}$-periodic measure $\mu^{u}\left(K\right)=u^{d}\mu\left(u^{-1}K\right)$.
Fourier coefficient of $\mu^{u}$ with index $u^{-1}\xi$, $\xi\in\mathbf{T}^{*}$, is $m_{\xi}$.

$\kappa_{d}=\pi^{\frac{d}{2}}\Gamma\left(\frac{d}{2}+1\right)^{-1}$
is volume of the unit ball in $\mathbf{\mathbb{R}}^{d}$.

We define the coefficient of $\mathbf{T}$-periodic measure $\mu$
as 
\begin{equation}
C_{\mu}=\frac{1}{2\pi^{2}d\kappa_{d}}\sum_{\xi\in\mathbf{T}^{*}}^{\xi\neq0}\frac{\left|m_{\xi}\right|^{2}}{\left|\xi\right|^{d+1}}\label{eq:coeff}.
\end{equation}

\begin{thm}
\label{thm:asymp}Let $f\in\mathbf{L}^{1}\left(\mathbf{\mathbb{R}}^{d}\right)\cap\mathbf{L}^{2}\left(\mathbf{\mathbb{R}}^{d}\right)$
and let $\mathfrak{g}_{f}'\left(0+\right)$, right derivative of its
isotropic covariogram in $0$, exist, let $\mu$ be $\mathbf{T}$-periodic
measure and $u>0$. Then the variance of $\widehat{Mf}\star\mu^{u}\left(x\right)$,
where $x$ and $M$ are uniform random, is
\[
\intop_{\mathbf{S}\mathbf{O}_{d}}\alpha\intop_{F_{\mathbf{T}}}\left|\widehat{Mf}\star\mu^{u}\left(x\right)-\lambda J_{f}\right|^{2}dxdp\left(M\right)=
\]
\[
=C_{\mu}\left(-\frac{d\kappa_{d}}{\kappa_{d-1}}\mathfrak{g}_{f}'\left(0+\right)\right)\Phi_{f}\left(u^{-1}\right)u^{d+1},
\]
where $J_{f}=\intop_{\mathbf{\mathbb{R}}^{d}}f\left(y\right)dy$ and
$\Phi_{f}$ is nonnegative function on $\mathbf{\mathbb{R}}^{+}$
such, that $\lim_{x\rightarrow+\infty}\frac{1}{x}\int_{0}^{x}\Phi_{f}\left(x\right)\:dx=1$.
\end{thm}

\begin{proof}
Mean value of $\widehat{Mf}\star\mu^{u}\left(x\right)$ is $\lambda\intop_{\mathbf{\mathbb{R}}^{d}}f\left(y\right)dy=\lambda J_{f}$
by Eq. \ref{eq:mean} for any $M\in\mathbf{S}\mathbf{O}_{d}$. Variance
decomposition theorem \citep{rao:1973} and Eq. \ref{eq:var-freq} yield
\[
\intop_{\mathbf{S}\mathbf{O}_{d}}\alpha\intop_{F_{\mathbf{T}}}\left|\widehat{Mf}\star\mu^{u}\left(x\right)-\lambda J_{f}\right|^{2}dxdp\left(M\right)=
\]
\[
=\sum_{\xi\in\mathbf{T}^{*}}^{\xi\neq0}\left|m_{\xi}\right|^{2}\mathcal{F}\mathfrak{g}_{f}\left(u^{-1}\left|\xi\right|\right).
\]

We set
\[
\Psi_{f}\left(y\right)=-\frac{2\pi^{2}\kappa_{d-1}y^{d+1}}{\mathfrak{g}_{f}'\left(0+\right)}\mathcal{F}\mathfrak{g}_{f}\left(y\right),
\]
\[
\Phi_{f}\left(x\right)=\sigma_{\mathbf{T}}^{-1}\sum_{\xi\in\mathbf{T}^{*}}^{\xi\neq0}c_{\xi}\Psi_{f}\left(\left|\xi\right|x\right),
\]
\[
c_{\xi}=\left|m_{\xi}\right|^{2}\left|\xi\right|^{-d-1},\;\sigma_{\mathbf{T}}=\sum_{\xi\in\mathbf{T}^{*}}^{\xi\neq0}c_{\xi},
\]
\[
L\left(u\right)=\frac{1}{2\pi^{2}\kappa_{d-1}}\left(d\kappa_{d}-2\pi u^{1-\frac{d}{2}}J_{\frac{d}{2}-1}\left(2\pi u\right)\right).
\]

Then the variance is
\[
\sum_{\xi\in\mathbf{T}^{*}}^{\xi\neq0}\left|m_{\xi}\right|^{2}\mathcal{F}\mathfrak{g}_{f}\left(u^{-1}\left|\xi\right|\right)=
\]
\[
=\frac{-\mathfrak{g}_{f}'\left(0+\right)}{2\pi^{2}\kappa_{d-1}}\left(\sum_{\xi\in\mathbf{T}^{*}}^{\xi\neq0}\frac{\left|m_{\xi}\right|^{2}}{\left|\xi\right|^{d+1}}\Psi_{f}\left(u^{-1}\left|\xi\right|\right)\right)u^{d+1}=
\]
\[
=\frac{-\mathfrak{g}_{f}'\left(0+\right)}{2\pi^{2}\kappa_{d-1}}\left(\sum_{\xi\in\mathbf{T}^{*}}^{\xi\neq0}\frac{\left|m_{\xi}\right|^{2}}{\left|\xi\right|^{d+1}}\right)\Phi_{f}\left(u^{-1}\right)u^{d+1}=
\]
\[
=C_{\mu}\left(-\frac{d\kappa_{d}}{\kappa_{d-1}}\mathfrak{g}_{f}'\left(0+\right)\right)\Phi_{f}\left(u^{-1}\right)u^{d+1}
\]
by Eq. \ref{eq:coeff}.

From definition of derivative it follows that
\[
\frac{-\mathfrak{g}_{f}'\left(0+\right)}{2\pi^{2}\kappa_{d-1}}=\frac{1}{2\pi^{2}\kappa_{d-1}}\lim_{h\rightarrow0+}\frac{1}{h}\left(\mathfrak{g}_{f}\left(0\right)-\mathfrak{g}_{f}\left(h\right)\right)=
\]
and by Eq. \ref{eq:Haenkel} and \ref{eq:cov-inv}
\[
=\lim_{h\rightarrow0+}\frac{1}{h}\intop_{0}^{\infty}L\left(h\rho\right)\rho^{d-1}\mathcal{F}\mathfrak{g}_{f}\left(\rho\right)\:d\rho=
\]
and applying the following
Proposition \ref{prop:ergodic} to the function $\rho^{d-1}\mathcal{F}\mathfrak{g}_{f}\left(\rho\right)$ we obtain finally
\[
=\lim_{R\rightarrow+\infty}\frac{1}{R}\intop_{0}^{R}\rho^{d+1}\mathcal{F}\mathfrak{g}_{f}\left(\rho\right)\:d\rho.
\]

Thus $\lim_{x\rightarrow+\infty}\frac{1}{x}\int_{0}^{x}\Psi_{f}\left(x\right)\:dx=1$,
$0\leq\Psi_{f}$ and it can be proved easily that $\lim_{x\rightarrow+\infty}\frac{1}{x}\int_{0}^{x}\Phi_{f}\left(x\right)\:dx=1$,
$0\leq\Phi_{f}$. 
\end{proof}
$\mathbf{M}^{1}\left(\mathbb{R}\right)$ is normed space of such continuous
functions $f$, that 
\[
\sum_{i=-\infty}^{\infty}\max\left\{ \left|f\left(x\right)\right|,\:i<x<i+1\right\} <\infty.
\]
Let $\psi$ have uniformly bounded variation on unit intervals in $\mathbf{\mathbb{R}}$, 
$N\in\mathbf{M}^{1}\left(\mathbb{R}\right)$, $\mathcal{F}N\left(\tau\right)\neq0$
for each $\tau\in\mathbb{R}$ and
\[
\lim_{\eta\rightarrow\infty}\intop_{-\infty}^{\infty}N\left(\eta-t\right)\:d\psi\left(t\right)=a\mathcal{F}N\left(0\right)
\]
then
\[
\lim_{\eta\rightarrow\infty}\intop_{-\infty}^{\infty}f\left(\eta-t\right)\:d\psi\left(t\right)=a\mathcal{F}f\left(0\right)
\]
for every $f\in\mathbf{M}^{1}\left(\mathbb{R}\right)$ by the Wiener second Tauberian theorem (\cite{wiener:1933}, Theorem 5).

Following Proposition \ref{prop:ergodic} was proved by Wiener (\cite{wiener:1933}, Theorem 21) in special case
corresponding to $d=1$.
\begin{prop}
\label{prop:ergodic}Let $f\left(r\right)\geq0$ and $f\in\mathbf{L}_{loc}^{1}\left(\mathbf{\mathbb{R}}^{+}\right)$.
Then 
\[
\lim_{R\rightarrow\infty}\frac{1}{R}\intop_{0}^{R}r^{2}f\left(r\right)\:dr=\lim_{h\rightarrow0+}\frac{1}{h}\intop_{0}^{\infty}L\left(hr\right)f\left(r\right)\:dr,
\]
where $L$ is defined in Theorem \ref{thm:asymp}, if and only if one of the
limits exists.
\end{prop}

\begin{proof}
We may set $f=0$ in interval $\left(0,1\right)$ according to the
theorem on monotone convergence (for all $R$ and for small $h$).

Then we set $r=$$\exp(t)$ and
\[
\begin{array}{cc}
R=\exp\left(\eta\right)=\frac{1}{h},\; & r^{2}f\left(r\right)=\varphi\left(t\right),\end{array}
\]
\[
N_{1}\left(s\right)=I_{\left\{ s|s>0\right\} }\exp\left(-s\right),
\]
\[
N_{2}\left(s\right)=\exp\left(s\right)L\left(\exp\left(-s\right)\right),
\]
\[
\psi\left(t\right)=\intop_{0}^{t}\varphi\left(s\right)\,ds.
\]

We can rewrite the proposition statement as:
\[
\lim_{\eta\rightarrow\infty}\intop_{-\infty}^{\infty}N_{1}\left(\eta-t\right)\:d\psi\left(t\right)=\lim_{\eta\rightarrow\infty}\intop_{-\infty}^{\infty}N_{2}\left(\eta-t\right)\:d\psi\left(t\right).
\]

$N_{i}\left(\eta\right)>0$ for $0\le\eta\le1$ and $i=1,2$, thus, if one of the above limits exists, then
\[
\limsup\intop_{n}^{n+1}d\psi\left(s\right)<\infty
\]
and also
\[
\intop_{n}^{n+1}d\psi\left(s\right)<M<\infty,
\]
i.e. the variation of $\psi$ on unit intervals is bounded by $M<\infty$.

We calculate Fourier transform of kernel $N_{2}$:
\[
\mathcal{F}N_{2}\left(\tau\right)=\intop_{0}^{\infty}L\left(u\right)u^{2\pi i\tau-2}du=
\]
using $\int t^{-\nu}J_{\nu+1}\left(t\right)dt=-t^{-\nu}J_{\nu}\left(t\right)$
\[
=\frac{2}{\sqrt{\pi}}\Gamma\left(\frac{d+1}{2}\right)\intop_{0}^{\infty}u^{2\pi i\tau-2}\intop_{0}^{u}\left(\pi s\right)^{1-\frac{d}{2}}J_{\frac{d}{2}}\left(2\pi s\right)ds\:du=
\]
by integration per partes follows
\[
=\frac{2\Gamma\left(\frac{d+1}{2}\right)}{\sqrt{\pi}\left(1-2\pi i\tau\right)}\left(-\left[u^{2\pi i\tau-1}\intop_{0}^{u}\frac{J_{\frac{d}{2}}\left(2\pi s\right)}{\left(\pi s\right)^{\frac{d}{2}-1}}ds\right]_{0}^{\infty}+\right.
\]
\[
\left.+\pi^{1-2\pi i\tau}\intop_{0}^{\infty}\frac{J_{\frac{d}{2}}\left(2\pi u\right)}{\left(\pi u\right)^{\frac{d}{2}-2\pi i\tau}}du\right),
\]
where the first term in parentheses is zero, because of l'Hospital rule
and asymptotics of Bessel function, and using
\[
\int_{0}^{\infty}t^{a}J_{\nu}\left(2t\right)dt=\frac{1}{2}\Gamma\left(\frac{\nu+a+1}{2}\right)\Gamma\left(\frac{\nu-a+1}{2}\right)^{-1},
\]
that holds for $Re\:a<\frac{1}{2}$, $Re\:a+\nu>-1$, we obtain finally:
\[
\mathcal{F}N_{2}\left(\tau\right)=\frac{\pi^{-2\pi i\tau-\frac{1}{2}}}{1-2\pi i\tau}\frac{\Gamma\left(\frac{d+1}{2}\right)\Gamma\left(\frac{1}{2}+\pi i\tau\right)}{\Gamma\left(\frac{d+1}{2}-\pi i\tau\right)},
\]
thus $\mathcal{F}N_{2}\left(0\right)=1$ and $\mathcal{F}N_{2}\left(\tau\right)\neq0$
for $\tau\in\mathbb{R}$, because Gamma function has no roots and
its poles are nonpositive integers.

The function $N_{1}$ is not continuous (thus $N_{1}\notin\mathbf{M}^{1}\left(\mathbb{R}\right)$).
We set 
\[
N_{1,\varepsilon}\left(t\right)=\frac{1}{\varepsilon}\int_{t}^{t+\varepsilon}N_{1}\left(s\right)\:ds
\]
so, that $N_{1,\varepsilon}\in\mathbf{M}^{1}\left(\mathbb{R}\right)$.
Then
\[
\mathcal{F}N_{1,\varepsilon}\left(\tau\right)=\frac{\exp\left(2\pi i\tau\varepsilon\right)-1}{2\pi i\tau\epsilon}\frac{1}{1+2\pi i\tau}
\]
has no real root and $\mathcal{F}N_{1,\varepsilon}\left(0\right)=1$.
Thus the limits are equal for $N_{1,\varepsilon}$ and $N_{2}$ if
one of the limits exists by second Wiener Tauberian theorem. Using
inequalities
\[
\frac{1-\exp\left(-\varepsilon\right)}{\varepsilon}N_{1}\left(t\right)\leq N_{1,\varepsilon}\left(t\right)\leq\frac{\exp\left(\varepsilon\right)-1}{\varepsilon}N_{1}\left(t\right)
\]
and the fact, that $\psi$ is nondecreasing and letting $\varepsilon\rightarrow0$
we finally see, that the limits with $N_{1}$ and $N_{2}$ are equal,
if one of the limits exists.
\end{proof}

Following propositions are based on the measure theory \citep{rudin:1987} and
on the theory of functions with bounded variation and of sets with
finite perimeter \citep{ziemer:1989}.

Let $K\subset\mathbf{\mathbb{R}}^{d}$ be measurable set and let the
vector measure $DI_{K}$ be the gradient of characteristic function
$I_{K}$ in the sense of distributions. The Radon-Nikodym theorem
yields the polar decomposition $DI_{K}=\nabla_{I_{K}}\left|DI_{K}\right|$
where the vector function $\nabla_{I_{K}}$ is inner normal and positive
measure $\left|DI_{K}\right|$ is variation of $DI_{K}$. Perimeter
of a measurable set $K\subset\mathbf{\mathbb{R}}^{d}$ is $\mathrm{Per}\left(K\right)=\left|DI_{K}\right|\left(\mathbf{\mathbb{R}}^{d}\right)$. 

If $K$ is measurable subset of $\mathbf{\mathbb{R}}^{d}$, $\mathrm{Per}\left(K\right)<\infty$
and $V\in\mathbf{C}_{0}^{1}\left(\mathbf{\mathbb{R}}^{d};\mathbf{\mathbb{R}}^{d}\right)$, then by the Gauss-Green theorem \citep{ziemer:1989}
\[
\intop_{K}\mathrm{div}V\,dx=-\intop_{\partial^{*}K}V\left(x\right)\cdot\nabla_{I_{K}}\left(y\right)\:dH^{d-1}\left(x\right),
\]
where $\nabla_{I_{K}}\left(y\right)$ is inner normal, $H^{d-1}$
is Hausdorf measure and $\partial^{*}K$ is set of points where the measure theoretic normal to set $K$ exists. 

Covariogram of $fI_{K}$ is
\[
g_{fI_{K}}\left(x\right)=\intop_{K\cap K-x}f\left(y+x\right)\overline{f\left(y\right)}dy
\]
and corresponding isotropic covariogram is 
\[
\mathfrak{g}_{fI_{K}}\left(\left|x\right|\right)=\int_{\mathbf{S}\mathbf{O}_{d}}g_{M\left(fI_{K}\right)}\left(x\right)dp\left(M\right).
\]
\begin{prop}
\label{prop:set-fce}Let $K\subset\mathbf{\mathbb{R}}^{d}$ be measurable
set. If $\mathrm{Per}\left(K\right)<\infty$ and $f\in\mathbf{C}_{c}^{1}\left(\mathbf{\mathbb{R}}^{d}\right)$,
then
\[
\mathfrak{g}_{fI_{K}}'\left(0+\right)=-\frac{\kappa_{d-1}}{d\kappa_{d}}\int_{\partial^{*}K}\left|f\right|^{2}\left(x\right)dH^{d-1}\left(x\right),
\]
where $\partial^{*}K$ is set of points where the measure theoretic normal to set $K$ exists and $H^{d-1}$ is Hausdorff
measure.
\end{prop}

\begin{proof}
Let $u\in S^{d-1}$ and $\varepsilon>0$, then
\[
\frac{g_{fI_{K}}\left(-\varepsilon u\right)-2g_{fI_{K}}\left(0\right)+g_{fI_{K}}\left(\varepsilon u\right)}{2\varepsilon}=
\]
\[
=-\frac{1}{2\varepsilon}\intop_{K\cap K-\varepsilon u}\left|f\left(y\right)-f\left(y+\varepsilon u\right)\right|^{2}dy-
\]
\[
-\frac{1}{2\varepsilon}\intop_{K\setminus K-\varepsilon u}\left|f\left(y\right)\right|^{2}dy-\frac{1}{2\varepsilon}\intop_{K\setminus K+\varepsilon u}\left|f\left(y\right)\right|^{2}dy
\]
first integral converges to $0$ because $f$ is Lipschitz and the
rest converges by the following Lemma \ref{lem:set-dir} to
\[
-\frac{1}{2}\intop_{\partial^{*}K}\left|f\right|^{2}\left(y\right)\left|u\cdot\nabla_{I_{K}}\left(y\right)\right|\:dH^{d-1}\left(y\right).
\]
The statement of the theorem then follows from the identity (\cite{galerne:2011}
Proposition 8)
\[
\frac{1}{d\kappa_{d}}\intop_{S^{d-1}}\left|u\cdot v\right|dH^{d-1}\left(u\right)=2\kappa_{d-1}\left\Vert v\right\Vert .
\]
\end{proof}
\begin{lem}
\label{lem:set-dir}Let $K\subset\mathbf{\mathbb{R}}^{d}$ be measurable
set, $\mathrm{Per}\left(K\right)<\infty$, $u\in S^{d-1}$, real function
$h\in\mathbf{C}_{c}^{1}\left(\mathbf{\mathbb{R}}^{d}\right)$ and $\varepsilon>0$,
then
\[
\frac{1}{\varepsilon}\intop_{K\setminus K-\varepsilon u}h\left(y\right)dy+\frac{1}{\varepsilon}\intop_{K\setminus K+\varepsilon u}h\left(y\right)dy
\]
 converges for $\varepsilon\rightarrow0$ to
\[
\intop_{\partial^{*}K}h\left(y\right)\left|u\cdot\nabla_{I_{K}}\left(y\right)\right|\:dH^{d-1}\left(y\right).
\]
\end{lem}
\begin{proof}
We show, that
\[
\frac{1}{\varepsilon}\intop_{K\setminus K-\varepsilon u}h\left(y\right)dy
\]
converges for $\varepsilon\rightarrow0$ to
\[
\intop_{\partial^{*}K}h\left(y\right)\left(u\cdot\nabla_{I_{K}}\left(y\right)\right)^{-}dH^{d-1}\left(y\right).
\]

Indeed
\[
\frac{1}{\varepsilon}\intop_{K\setminus K-\varepsilon u}h\left(y\right)dy-\frac{1}{\varepsilon}\intop_{K-\varepsilon u\setminus K}h\left(y\right)dy=
\]
\[
=\frac{1}{\varepsilon}\intop_{K}h\left(y\right)-h\left(y-\varepsilon u\right)dy=
\]
\[
=\frac{1}{\varepsilon}\intop_{K}u\cdot DI_{K}\intop_{-\varepsilon}^{0}h\left(y+tu\right)dt\,dy=
\]
and the Gauss-Green theorem with
\[
V\left(x\right)=\frac{1}{\varepsilon}\intop_{-\varepsilon}^{0}h\left(x+tu\right)dt\,u
\]
gives
\[
=-\intop_{\partial^{*}K}\frac{1}{\varepsilon}\intop_{-\varepsilon}^{0}h\left(y+tu\right)dt\,u\cdot\nabla_{I_{K}}\left(y\right)dH^{d-1}\left(y\right).
\]

Let $N_{\varepsilon u}=\varepsilon^{-1}H^{1}|\left\langle 0,\varepsilon u\right\rangle $
and measure
\[
\mu=-u\cdot\nabla_{I_{K}}H^{d-1}|\partial^{*}K.
\]
We have shown, that 
\[
\frac{1}{\varepsilon}\intop_{K\setminus K-\varepsilon u}h\left(y\right)dy-\frac{1}{\varepsilon}\intop_{K-\varepsilon u\setminus K}h\left(y\right)dy=N_{\varepsilon u}\star\mu\left(h\right).
\]
Following Lemma \ref{lem:weak} gives $N_{\varepsilon u}\star\mu\stackrel{w^{*}}{\rightarrow}\mu$
and $\left(N_{\varepsilon u}\star\mu\right)^{+}\stackrel{w^{*}}{\rightarrow}\mu^{+}$.
By Radon-Nikodym theorem
\[
\mu^{+}=\left(u\cdot\nabla_{I_{K}}\right)^{-}H^{d-1}|\partial^{*}K,
\]
thus
\[
\frac{1}{\varepsilon}\intop_{K\setminus K-\varepsilon u}h\left(y\right)dy=\left(N_{\varepsilon u}\star\mu\right)^{+}\left(h\right)\rightarrow\mu^{+}\left(h\right)
\]
what was to be shown.

Similarly
\[
\frac{1}{\varepsilon}\intop_{K\setminus K+\varepsilon u}h\left(y\right)dy==\left(N_{\varepsilon u}\star\mu\right)^{-}\left(h\right)
\]
converges for $\varepsilon\rightarrow0$ to
\[
\intop_{\partial^{*}K}h\left(y\right)\left(u\cdot\nabla_{I_{K}}\left(y\right)\right)^{+}dH^{d-1}\left(y\right)=\mu^{-}\left(h\right).
\]
\end{proof}
\begin{lem}
\label{lem:pos}Let $\nu$ be real Borel measure, $f\in\mathbf{C}_{c}\left(\mathbf{\mathbb{R}}^{d}\right)$,
$f\geq0$, then 
\[
\nu^{+}\left(f\right)=\sup\left\{ \nu\left(g\right),\:0\leq g\leq f,\:g\in\mathbf{C}_{c}\left(\mathbf{\mathbb{R}}^{d}\right)\right\} .
\]
\end{lem}
\begin{proof}
Radon-Nikodym theorem gives measurable function $p$ such, that $\nu=p\left|\nu\right|$.
Obviously
\[
\nu^{+}=p^{+}\left|\nu\right|=p^{+}\nu.
\]
Luzin theorem
provides $p_{\varepsilon}\in\mathbf{C}_{c}\left(\mathbf{\mathbb{R}}^{d}\right)$
such, that $\left|\nu\right|\left(p^{+}-p_{\varepsilon}\right)<\varepsilon$
and $0\leq p_{\varepsilon}\leq1$ for each $\varepsilon>0$. Then
\[
\nu\left(p_{\varepsilon}f\right)=\nu\left(p_{\varepsilon}f-p^{+}f\right)+\nu\left(p^{+}f\right)\geq\nu^{+}\left(f\right)-\varepsilon\sup f
\]
and
\[
\nu^{+}\left(f\right)\leq\underset{0\leq g\leq f}{\sup}\nu\left(g\right).
\]

On the other hand $\nu^{+}\left(f\right)\geq\nu^{+}\left(g\right)\geq\nu\left(g\right)$,
thus
\[
\nu^{+}\left(f\right)\geq\underset{0\leq g\leq f}{\sup}\nu\left(g\right).
\]
\end{proof}
\begin{lem}
\label{lem:weak}Let $\mu$ be real Borel measure on $\mathbf{\mathbb{R}}^{d}$
and $N_{\varepsilon u}=\varepsilon^{-1}H^{1}|\left\langle 0,\varepsilon u\right\rangle $.
Then $N_{\varepsilon u}\star\mu\stackrel{w^{*}}{\rightarrow}\mu$ and
$\left(N_{\varepsilon u}\star\mu\right)^{+}\stackrel{w^{*}}{\rightarrow}\mu^{+}$
for $\varepsilon\rightarrow0+$ .
\end{lem}

\begin{proof}
Let $h\in\mathbf{C}_{c}\left(\mathbf{\mathbb{R}}^{d}\right)$ and
$N_{\varepsilon u}\star\mu\left(h\right)=\mu\left(\widehat{N_{\varepsilon u}}\star h\right)$
where $\widehat{N_{\varepsilon u}}=\varepsilon^{-1}H^{1}|\left\langle -\varepsilon u,0\right\rangle $.
Since $h$ is uniformly continuous $\widehat{N_{\varepsilon u}}\star h$
converges to $h$ for $\varepsilon\rightarrow0$ uniformly, thus $N_{\varepsilon u}\star\mu\stackrel{w^{*}}{\rightarrow}\mu$
.

We proof the second statement using $h\in\mathbf{C}_{c}\left(\mathbf{\mathbb{R}}^{d}\right)$
and $h\geq0$ as the general case follows from decomposition $h=h^{+}-h^{-}$
and linearity. Lemma \ref{lem:pos} gives 
\[
\left(N_{\varepsilon u}\star\mu\right)^{+}\left(h\right)=\underset{0\leq g\leq h}{\sup}N_{\varepsilon u}\star\mu\left(g\right)=\underset{0\leq g\leq h}{\sup}\mu\left(\widehat{N_{\varepsilon u}}\star g\right)\leq
\]
and from $0\leq g\leq h$ and $\widehat{N_{\varepsilon u}}>0$ follows
$0\leq\widehat{N_{\varepsilon u}}\star g\leq\widehat{N_{\varepsilon u}}\star h$, thus
\[
\leq\underset{0\leq\widehat{N_{\varepsilon u}}\star g\leq\widehat{N_{\varepsilon u}}\star h}{\sup}\mu\left(\widehat{N_{\varepsilon u}}\star g\right)\leq
\]
and Lemma \ref{lem:pos} yields 
\[
\leq\underset{0\leq f\leq\widehat{N_{\varepsilon u}}\star h}{\sup}\mu\left(f\right)=\mu^{+}\left(\widehat{N_{\varepsilon u}}\star g\right).
\]
From identity $\mu^{+}\left(\widehat{N_{\varepsilon u}}\star g\right)=N_{\varepsilon u}\star\mu^{+}\left(g\right)$
finally follows 
\[
\left(N_{\varepsilon u}\star\mu\right)^{+}\left(h\right)\leq N_{\varepsilon u}\star\mu^{+}\left(h\right).
\]
 $\widehat{N_{\varepsilon u}}\star h$ uniformly converges to $h$, thus
\[
N_{\varepsilon u}\star\mu^{+}\left(h\right)\rightarrow\mu^{+}\left(h\right),
\]
i.e. for each $\delta>0$ there is $\varepsilon_{0}$ such, that if
$\varepsilon<\varepsilon_{0}$ then
\[
N_{\varepsilon u}\star\mu^{+}\left(h\right)<\mu^{+}\left(h\right)+\delta,
\]
thus
\[
\limsup\left(N_{\varepsilon u}\star\mu\right)^{+}\left(h\right)\leq\mu^{+}\left(h\right).
\]

The opposite inequality follows from Lemma \ref{lem:pos}: for $\delta>0$
there is $f\in\mathbf{C}_{c}\left(\mathbf{\mathbb{R}}^{d}\right)$
such, that $0\leq f\leq h$ and $\mu\left(f\right)>\mu^{+}\left(h\right)-\delta$.
$f$ is uniformly continuous, thus $N_{\varepsilon u}\star\mu\left(f\right)\rightarrow\mu\left(f\right)$
and there is $\varepsilon_{0}$ such, that if $\varepsilon<\varepsilon_{0}$
then
\[
N_{\varepsilon u}\star\mu\left(f\right)>\mu\left(f\right)-\delta.
\]
Thus $N_{\varepsilon u}\star\mu\left(f\right)>\mu^{+}\left(h\right)-2\delta$
for arbitrary $\delta>0$ and
\[
\liminf\left(N_{\varepsilon u}\star\mu\right)^{+}\left(h\right)\geq\mu^{+}\left(h\right).
\]
\end{proof}
\begin{thm}
\label{thm:main}Let $K\subset\mathbf{\mathbb{R}}^{d}$ be measurable set. If $\mathrm{Per}\left(K\right)<\infty$, 

i) then
\[
\intop_{\mathbf{S}\mathbf{O}_{d}}\alpha\intop_{F_{\mathbf{T}}}\left|\mu^{u}\left(MK+x\right)-\lambda\left|K\right|\right|^{2}dxdp\left(M\right)=
\]
\[
=C_{\mu}\mathrm{Per}\left(K\right)\Phi_{K}\left(u^{-1}\right)u^{d+1},
\]

ii) and if $f\in\mathbf{C}_{c}^{1}\left(\mathbf{\mathbb{R}}^{d}\right)$,
then
\[
\intop_{\mathbf{S}\mathbf{O}_{d}}\alpha\intop_{F_{\mathbf{T}}}\left|\widehat{MfI_{K}}\star\mu^{u}\left(x\right)-\lambda J_{f}\right|^{2}dxdp\left(M\right)=
\]
\[
=C_{\mu}\left(\int_{\partial^{*}K}\left|f\right|^{2}\left(x\right)dH^{d-1}\left(x\right)\right)\Phi_{fI_{K}}\left(u^{-1}\right)u^{d+1},
\]

iii) and if $f_{1},f_{2}\in\mathbf{C}_{c}^{1}\left(\mathbf{\mathbb{R}}^{d}\right)$,
then 
\[
\intop_{\mathbf{S}\mathbf{O}_{d}}\alpha\intop_{F_{\mathbf{T}}}\left(\widehat{Mf_{1}I_{K}}\star\mu^{u}\left(x\right)-\lambda J_{f_{1}}\right)\cdot\
\]
\[
\cdot\overline{\left(\widehat{Mf_{2}I_{K}}\star\mu^{u}\left(x\right)-\lambda J_{f_{2}}\right)}dxdp\left(M\right)=
\]
\[
=C_{\mu}\left(\mathrm{Re}\int_{\partial^{*}K}f_{1}\left(x\right)\overline{f_{2}\left(x\right)}dH^{d-1}\left(x\right)\right)\cdot\
\]
\[
\cdot\Phi_{f_{1}I_{K},f_{2}I_{K}}\left(u^{-1}\right)u^{d+1},
\]
 where $\Phi_{\cdot}$ is nonnegative function on $\mathbf{\mathbb{R}}^{+}$
such, that $\lim_{x\rightarrow+\infty}\frac{1}{x}\int_{0}^{x}\Phi_{\cdot}\left(x\right)\:dx=1$.
\end{thm}

\begin{proof}
i) follows from Theorem \ref{thm:asymp} and equation $-d\kappa_{d}\mathfrak{g}_{K}'\left(0+\right)=\kappa_{d-1}\mathrm{Per}\left(K\right)$
(\cite{galerne:2011}, Theorem 14).

ii) follows from Theorem \ref{thm:asymp} and from Proposition \ref{prop:set-fce}.

iii) Cross-covariogram of functions $f_{1},f_{2}$
to the set $K$ is
\[
g_{f_{1}I_{K},f_{2}I_{K}}\left(x\right)=\intop_{K\cap K-x}f_{1}\left(y+x\right)\overline{f_{2}\left(y\right)}dy
\]
and corresponding isotropic cross-covariogram is
\[
\mathfrak{g}_{f_{1}I_{K},f_{2}I_{K}}\left(\left|x\right|\right)=\int_{\mathbf{S}\mathbf{O}_{d}}g_{M\left(f_{1}I_{K}\right),M\left(f_{2}I_{K}\right)}\left(x\right)dp\left(M\right),
\]
obviously 
\[
\mathfrak{g}_{f_{1}\pm f_{2}I_{K}}=\mathfrak{g}_{f_{1}I_{K}}\pm2\mathfrak{g}_{f_{1}I_{K},f_{2}I_{K}}+\mathfrak{g}_{f_{2}I_{K}}.
\]
From
\[
\int_{\partial^{*}K}\left|f_{1}\left(x\right)\pm f_{2}\left(x\right)\right|^{2}dH^{d-1}\left(x\right)=
\]
\[
=\int_{\partial^{*}K}\left|f_{1}\right|^{2}\left(x\right)dH^{d-1}\left(x\right)\pm
\]
\[
\pm2\mathrm{Re}\int_{\partial^{*}K}f_{1}\left(x\right)\overline{f_{2}\left(x\right)}dH^{d-1}\left(x\right)+
\]
\[
+\int_{\partial^{*}K}\left|f_{2}\right|^{2}\left(x\right)dH^{d-1}\left(x\right)
\]
and from Proposition \ref{prop:set-fce} follows 
\[
-d\kappa_{d}\mathfrak{g}_{f_{1}I_{K},f_{2}I_{K}}'\left(0+\right)=
\]
\[
=\kappa_{d-1}\mathrm{Re}\int_{\partial^{*}K}f_{1}\left(x\right)\overline{f_{2}\left(x\right)}dH^{d-1}\left(x\right),
\]
where $\partial^{*}K$ is set of points where the measure theoretic normal to set $K$ exists and $H^{d-1}\left(x\right)$
Hausdorff measure. The statement then follows from ii) applied to
$\left(f_{1}\pm f_{2}\right)I_{K}$. 
\end{proof}

\section{Results}
\subsection{Verification of variance formula and test of parameters}

General ellipsoid was measured with sevenfold grid with random orientation and position repeatedly. The variance of semiaxes length was calculated from measurement or estimated by Eq. \ref{eq:semivar}. Results presented in Tab. 1 show excellent performance of the error estimation formula.

\begin{center}
\begin{tabular}{|c|c|c|c|}
\hline 
Ellipsoid (n=100) & $s_{1}$ & $s_{2}$ & $s_{3}$\tabularnewline
\hline 
\hline 
Simulation  & 0.34 & 0.30 & 0.24\tabularnewline
\hline 
Theoretical value & 0.33 & 0.28 & 0.23\tabularnewline
\hline 
\end{tabular}
\par\end{center}

Table 1. \textit{Standard deviation of estimate of ellipsoid semiaxes ($s_{1}$=50, $s_{2}$=40, $s_{3}$=30, in arbitrary units) and mean values of standard deviation (calculated using Eq. \ref{eq:semivar}) are presented. Standard deviation of estimated standard deviation was 0.01 in all semiaxes. Estimate by Fakir sevenfold grid ($L_{V}$=0.01183) was repeated 100x with random grid orientation and position.}

Error of semiaxes length estimate in two selected samples is presented in Tab. 2. Precision of the method with selected parameters (grid type and $L_{V}$) and given objects (adult male phaesant forebrain and hatchling forebrain) is better than 0.5 percent.

\begin{center}
\begin{tabular}{|c|c|c|c|}
\hline 
Forebrain & $s_{1}$ & $s_{2}$ & $s_{3}$\tabularnewline
\hline 
\hline 
Adult & 12.16 (0.03) & 7.95 (0.02) & 6.48 (0.02)\tabularnewline
\hline 
Hatch. & 7.18 (0.03) & 5.21 (0.02) & 4.76 (0.02)\tabularnewline
\hline 
\end{tabular}
\par\end{center}

Table 2. \textit{Semiaxes estimate ($s_{1}\geq s_{2}\geq s_{3}$, in $mm$) of selected samples (forebrain of adult male and hatchling) by Fakir sevenfold grid ($L_{V}$ = 0.76 $mm^{-2}$) are presented with the estimated standard deviation (calculated using Eq. \ref{eq:semivar})
in parenthesis.}

\subsection{Analysis of the developing pheasant brain}
Heads of 2 hatchlings, 4 juvenile and 6 adult ring-necked pheasants (Phasianus colchicus) were fixed in formalin before scanning. The images of brain were acquired at high resolution (voxel volume = 0.002775 $mm^{3}$) using a 4.7 T magnetic resonance (MR) spectrometer (Bruker BioSpec) equipped with a commercially available resonator coil, and 3D Rapid Acquisition incorporating a Relaxation Enhancement (RARE) multi-spin echo sequence \citep{jirak:2015}.
MR images were analyzed in a home-made Fakir software.
To achieve high precision for calculation of volume and surface area of brain divisions, they were measured interactively by sevenfold Fakir probe with grid density 0.76 $mm^{-2}$ and the results can be found in \citep{jirak:2015}. Structures of avian brain were identified in histological atlas \citep{karten:2013}. 

\begin{center}
\begin{tabular}{|c|c|c|c|}
\hline 
 & Forebrain & Midbrain & Hindbrain\tabularnewline
\hline 
\hline 
Hatch. (2) & 0.30 (0.01) & 0.83 (0.01) & 0.39 (0.02)\tabularnewline
\hline 
Juv. (4) & 0.36 (0.01) & 0.85 (0.01) & 0.44 (0.01)\tabularnewline
\hline 
Adults (6) & 0.43 (0.01) & 0.83 (0.01) & 0.36 (0.02)\tabularnewline
\hline 
\end{tabular}
\par\end{center}

Table 3. \textit{Procrustes anisotropy PA of Pheasant brain compartments (hatchlings, juveniles and adults, number of samples is in parentheses) calculated according to Eq. \ref{eq:procr}. Anisotropy mean values are presented with the
standard error of the mean. The differences in forebrain and hindbrain
anisotropy between age groups are statistically significant (ANOVA
$p<0.01$).}

\begin{center}
\begin{tabular}{|c|c|c|c|}
\hline 
a) Forebrain & $s_{1}$ & $s_{2}$ & $s_{3}$\tabularnewline
\hline 
\hline 
Hatch. (2) & 7.0 (0.0) & 5.2 (0.1) & 4.5 (0.1)\tabularnewline
\hline 
Juv. (4) & 10.3 (0.2) & 7.2 (0.1) & 6.3 (0.1)\tabularnewline
\hline 
Adults (6) & 11.7 (0.1) & 7.8 (0.1) & 6.3 (0.1)\tabularnewline
\hline 
\end{tabular}
\par\end{center}

\begin{center}
\begin{tabular}{|c|c|c|c|}
\hline 
b) Midbrain & $s_{1}$ & $s_{2}$ & $s_{3}$\tabularnewline
\hline 
\hline 
Hatch. (2) & 8.4 (0.1) & 3.3 (0.1) & 2.2 (0.0)\tabularnewline
\hline 
Juv. (4) & 10.7 (0.1) & 4.1 (0.1) & 3.0 (0.1)\tabularnewline
\hline 
Adults (6) & 11.1 (0.2) & 4.1 (0.1) & 3.6 (0.1)\tabularnewline
\hline 
\end{tabular}
\par\end{center}

\begin{center}
\begin{tabular}{|c|c|c|c|}
\hline 
c) Hindbrain & $s_{1}$ & $s_{2}$ & $s_{3}$\tabularnewline
\hline 
\hline 
Hatch. (2) & 5.1 (0.0) & 3.4 (0.1) & 3.0 (0.2)\tabularnewline
\hline 
Juv. (4) & 7.4 (0.0) & 4.6 (0.1) & 4.2 (0.0)\tabularnewline
\hline 
Adults (6) & 7.5 (0.1) & 5.2 (0.2) & 4.5 (0.1)\tabularnewline
\hline 
\end{tabular}
\par\end{center}

Table 4. \textit{(a) forebrain (b) midbrain and (c) hindbrain of (hatchlings, juveniles and adults, number of samples is in parentheses) volume tensor.
Semiaxes lengths ($s_{1}\geq s_{2}\geq s_{3}$, in $mm$) of ellipsoids
calculated according to Eq. \ref{eq:semi}, mean values are presented with the
standard error of the mean.}

The semi-major axis is oriented laterally in forebrain and midbrain and rostrally in hindbrain. The significant increase of PA in forebrain (Tab. 3) means, that the increase in the length of semi-major axis oriented laterally is more pronounced that the length of other semi-axes (Tab. 4a), i.e. the change of relative width is the major change of the forebrain shape during the development. In the same manner we conclude that the rostral elongation of hindbrain significantly increases during brain development.
Calculating the variance of the semi-axes estimate for selected brain compartments and grid density 0.76 $mm^{-2}$ using Eq. \ref{eq:semivar} shows very high precision; the error of the estimate of the semi-axes for given grid density and objects size is less than 0.4 percent. 

\section{Discussion}
Semi-axes provide only the most basic object shape characteristics. On the other hand, the estimate of linear dimensions of a 3D object obtained by the volume tensor method may be both more precise and more robust compared to direct measurement because the latter depends on the selection of extreme points within the object that may be rather arbitrary.

Fakir probe enables estimation of volume tensor of single object from 3D data without its explicit segmentation by sparse systematic sampling that makes the method efficient. Prediction of the method precision follows classical works \citep{hlawka:1950,matheron:1965} but it had to be proved de novo using Wiener-Tauberian and geometric measure theory. We applied the special arrangement of line grids in 7 fold grid in this study to increase the precision of estimate of surface integrals in Eq. \ref{eq:covasymp} by negative covariance between estimates using sets of parallel lines in different directions.

The method might be preferred to surface based method \citep{schroederturk:2011} in situation when the automatic segmentation of the object is not feasible. The method allows measurement of shape of pheasant brain compartments efficiently and we were able to detect changes in the shape of phaesant forebrain and hindbrain during development (Tab 3,4), which may be of importance for interpretation of fossilized braincases of extinct birds and dinosaurs. Our approach using Fakir grid for volume measurement can be used in similar morphometric studies of macro- or microscopic objects.
 
The Fakir probe implemented in MS Visual C++ including calculation of semi-axes and precision can be downloaded from the authors webpage.

\section{Acknowledgment}
This work was supported by Czech Science Foundation grant No. P302/12/1207, by MEYS (project LM2015062 Czech-BioImaging, project CZ.02.1.01/0.0/0.0/16\_013/0001775 Moderni-zation and support of research activities of the national infrastructure for biological and medical imaging Czech-BioImaging funded by OP RDE) and MH CR-DRO (Institute for Clinical and Experimental Medicine IKEM, IN00023001).

\end{paper}

\end{document}